\documentclass[12pt]{article}
\usepackage{amsmath}
\usepackage{amsfonts}
\usepackage{amssymb}
\usepackage{amsthm}
\usepackage{cite}
\usepackage[unicode=true,pdfusetitle,
 bookmarks=true,bookmarksnumbered=false,bookmarksopen=false,
 breaklinks=true,pdfborder={0 0 0},backref=false,colorlinks=true]
 {hyperref}
\usepackage{mathrsfs}
\usepackage{subcaption}
\usepackage{cancel}

\usepackage{color}
\newcommand{\bc}{\color{black}}
\newcommand{\bcc}{\color{black}}

\newlength{\dinwidth}
\newlength{\dinmargin}
\DeclareMathOperator{\im}{Im}

\newif\ifshowdetails
\showdetailstrue

\hypersetup{
 linkcolor=blue,citecolor=blue, urlcolor=blue}

\newcounter{todo}

\newcommand{\todobox}[1]{
  \textcolor{blue}{
    \fbox{\parbox{0.9\textwidth}{#1}}%
  }
}

\newcommand{\todonotetag}{TODO\thetodo}
\newcommand{\todonote}[1]{%
\stepcounter{todo}%
{\let\thefootnote\todonotetag
\footnote{\todobox{#1}}%
}
}

\definecolor{detailsgray}{gray}{0.3}

\ifshowdetails
\newcommand{\detail}[1]{%
{\color{detailsgray}$\blacktriangleright${#1}$\blacktriangleleft$}%
}

\newcommand{\detailspar}[1]{
\par \noindent {\color{detailsgray} $\blacktriangleright$ \textit{#1} $\blacktriangleleft$ } \par
}

\newenvironment{details}{\par \noindent \begingroup\itshape\color{detailsgray} $\blacktriangleright$ \hrulefill%
\par \noindent \ignorespaces}{\par\noindent \hrulefill $\blacktriangleleft$  \endgroup \par \noindent}

\else
\newcommand{\detail}[1]{} 
\newcommand{\detailspar}[1]{} 
\usepackage{environ}
\NewEnviron{hide}{}

\fi

\newcommand{\lam}{\theta}

\newcommand{\alphan}{\alpha}
\newcommand{\vh}{\mathbf{h}}

\newcommand{\vfe}{\vf_{\mrm{e}}}
\newcommand{\vfb}{\vf_{\mrm{b}}}

\usepackage{tikz}
\usetikzlibrary{trees}
\usetikzlibrary{decorations.pathmorphing}
\usetikzlibrary{decorations.markings}

\newcommand{\hvk}{\hat{\vk}}
\newcommand{\alt}{\tilde{\al}}

\newcommand{\pol}{\pmb{\epsilon}_{\la}}

\newcommand{\vv}{\pmb{v}}

\newcommand{\mrm}{\mathrm}
\newcommand{\vP}{\mathbf{P}}

\newcommand{\vA}{\mathbf{A}}

\newcommand{\vf}{\mathbf{f}}

\renewcommand{\mathbf}{\boldsymbol}

\newcommand{\mcF}{\mathcal F}

\newcommand{\mcL}{\mathcal L}

\numberwithin{equation}{section}

\newcommand{\mcS}{\mathcal S}

\newcommand{\wt}{\widetilde}

\newcommand{\vx}{\boldsymbol{x}}

\newcommand{\ti}{\tilde}

\newcommand{\vac}{\mrm{vac} }
\newcommand{\Om}{\Omega}

\newcommand{\ka}{\kappa}

\newcommand{\vk}{\boldsymbol{k}}

\newcommand{\be}{\beta}

\newcommand{\pa}{\partial}

\newcommand{\ov}{\overline}
\newcommand{\mfh}{\mathfrak{h}}

\newcommand{\eps}{\varepsilon}
\newcommand{\de}{\delta}

\newcommand{\De}{\Delta}

\newcommand{\pho}{\mathrm{ph}}

\newcommand{\nin}{\noindent}
\newcommand{\si}{\sigma}
\newcommand{\ph}{\phantom}
\newcommand{\h}{\fr{1}{2}}

\newcommand{\hil}{\mathcal{H}}
\newcommand{\om}{\omega}
\newcommand{\mfa}{\mathfrak{A}}
\newcommand{\mco}{\mathcal{O}}

\newcommand{\supp}{\mathrm{supp}}
\newcommand{\fr}[2]{\frac{#1}{#2}}
\newcommand{\al}{\alpha}
\newcommand{\real}{\mathbb{R}}
\newcommand{\complex}{\mathbb{C}}
\newcommand{\la}{\lambda}
\newcommand{\non}{\nonumber}

\newcommand{\lan}{\langle}
\newcommand{\ran}{\rangle}

\newcommand{\vp}{\mathbf{v}_{\mathbf{P}}}
\newcommand{\vptT}{\hat{\mathbf{v}}_{\mathbf{P},T}}
\newcommand{\vpt}{\hat{\mathbf{v}}_{\mathbf{P}}}
\newcommand{\fb}{\mathbf{f}}
\newcommand{\kb}{\mathbf{k}}
\newcommand{\Pb}{\mathbf{P}}
\newcommand{\km}{\lvert \mathbf{k} \rvert}

\newcommand{\vps}{\mathbf{v}_{\mathbf{P},\sigma}}

\def\qed{$\Box$\medskip}

\newtheorem{theoreme}{Theorem } [section]
\newtheorem{proposition}[theoreme]{Proposition}
\newtheorem{lemma}[theoreme]{Lemma}
\newtheorem{definition}[theoreme]{Definition}
\newtheorem{corollary}[theoreme]{Corollary}
\newtheorem{remark}[theoreme]{Remark}
\newtheorem{example}[theoreme]{Example}
\newtheorem{criterion}[theoreme]{Criterion}
\newtheorem{conjecture}{Conjecture}
\newtheorem{assumption}{Assumption}

\newcommand{\tr}{\mrm{tr}}

\newcommand{\bea}{\begin{assumption}}
	\newcommand{\eea}{\end{assumption}}
\newcommand{\beco}{\begin{conjecture} }
	\newcommand{\eeco}{\end{conjecture} }
\newcommand{\beq}{\begin{equation}}
	\newcommand{\eeq}{\end{equation}}
\newcommand{\beqa}{\begin{eqnarray}}
	\newcommand{\eeqa}{\end{eqnarray}}
\newcommand{\ben}{\begin{arabicenumerate}}
	\newcommand{\een}{\end{arabicenumerate}}
\newcommand{\bex}{\begin{example}}
	\newcommand{\eex}{\end{example}}
\newcommand{\ber}{\begin{remark}}
	\newcommand{\eer}{\end{remark}}
\newcommand{\bec}{\begin{corollary}}
	\newcommand{\eec}{\end{corollary}}
\newcommand{\bep}{\begin{proposition}}
	\newcommand{\eep}{\end{proposition}}
\newcommand{\becr}{\begin{criterion}}
	\newcommand{\eecr}{\end{criterion}}

\def\bel{\begin{lemma}}
	\def\eel{\end{lemma}}
\def\bet{\begin{theoreme}}
	\def\eet{\end{theoreme}}
\def\bed{\begin{definition}}
	\def\eed{\end{definition}}

\title{Curing velocity superselection in non-relativistic QED \\ by restriction to a lightcone}

 \author{
{\bf Daniela Cadamuro}\\
Institute for Theoretical Physics, University of Leipzig, \\
E-mail: {\tt daniela.cadamuro@itp.uni-leipzig.de }
 \and
{\bf Wojciech Dybalski} \\
Zentrum Mathematik, Technische Universit\"at M\"unchen,\\
E-mail: {\tt dybalski@ma.tum.de}
}

 \date{}

\begin{document}
\maketitle

\begin{abstract}
It is physically expected that plane-wave configurations of the electron in QED induce disjoint
representations of the algebra of the electromagnetic fields. This phenomenon of velocity superselection,
which is one aspect of the infrared problem, is mathematically well established in non-relativistic
(Pauli-Fierz type) models of QED. We show that velocity superselection can be resolved in
such models  by restricting the electron states to the  subalgebra of the  fields localized in the
future lightcone.  
This actually follows from a more general statement about equivalence of GNS representations for coherent states of the algebra of the {\bc future} lightcone in free electromagnetism.
Our analysis turns out to be meaningful in the non-relativistic setting
and provides evidence in favour of the Buchholz-Roberts approach to infrared problems. 
\end{abstract}

\section{ Introduction} 
In the framework of local relativistic QFT  D.~Buchholz and J.~E.~Roberts  proposed a novel approach to infrared problems, by focusing attention on measurements performed in some future lightcone  \cite{BR14}.
They defined a family of charged representations, localizable
in certain subsets of the future lightcone, and developed for them a meaningful superselection theory in the spirit of the Doplicher-Haag-Roberts (DHR)
analysis. As the Buchholz-Roberts approach invalidates the standard no-go theorems \cite{Bu86}, also a resolution
of the infraparticle problem, i.e., a demonstration of a sharp mass-shell for the electron, was posed as a question for future research in \cite{BR14}. 
It was later shown by S. Alazzawi and one of the present authors in \cite{AD15} that in the absence of the infraparticle
problem one
can construct Compton scattering states in the Buchholz-Roberts representations of QED. 
However, the question of a sharp mass of the electron was not addressed  in this work and 
it appears to be too specific to tackle it in the axiomatic setting. On the other hand, concrete non-perturbative models
of QED, amenable to a rigorous mathematical treatment, are non-relativistic due to severe ultraviolet problems.
As the algebra of observables localized in a lightcone is a priori not available in such models, they may not appear
suitable to test the Buchholz-Roberts approach. It is the goal of the present paper to show that such {\bcc a} conclusion would
in fact be pre-mature. We consider the well-established property of velocity superselection in {\bc non-relativistic QED},  which 
says that plane-wave configurations of the electron with distinct velocities induce disjoint representations of the
algebra of the electromagnetic fields. We show that a restriction to the subalgebra of the future lightcone is meaningful
in this context and that the phenomenon of velocity superselection disappears after such restriction. This means that 
the plane-wave configurations
become coherent and can, in principle, be superposed into normalisable states of the electron with sharp mass. However,
this latter step is not considered in this work.    

Let us explain in non-technical terms how velocity superselection is defined in models of non-relativistic QED and how we resolve it by restriction to a lightcone.
The Hilbert space of the model is $\hil=L^2(\real^3)\otimes \mcF_{\pho}$, where $L^2(\real^3)$ carries the  degrees of freedom
of a spinless electron and   $\mcF_{\pho}$ is the Fock space of the physical photon states.  The Hamiltonian has the textbook form 
(cf.~\cite{Sp})
\beqa
H:= \frac{1}{2}(-i\nabla_{\pmb{x}} + \ti{\alpha}^{1/2} \pmb{A}(\pmb{x}))^2 + H_{\pho},  \label{intro-Hamiltonian}
\eeqa
where $\ti{\al}>0$ is the coupling constant, $\pmb{x}$ is the position of the electron, $\pmb{A}$ is the electromagnetic potential in the Coulomb gauge with fixed ultraviolet regularization and $H_{\pho}$ is the  Hamiltonian of free photons.  Due to the translation invariance, we can decompose  
$H$ into the fiber Hamiltonians $H_{\pmb{P}}$ at fixed momentum $\vP$:
\begin{equation}
 H =\Pi^* \bigg(\int^{\oplus} H_{\pmb{P}} \; d^3 \pmb{P}\bigg)\Pi, \label{fiber-hamiltonians}
\end{equation} 
where $\Pi$ is a certain unitary map. 
The Hamiltonians  $H_{\pmb{P}}$, given by (\ref{fiber-Hamiltonians}) below, are self-adjoint operators acting on  the so called fiber Fock space which we denote by  $\mcF$. A manifestation of the infraparticle problem in this model is the absence of the ground states of  $H_{\pmb{P}}$, which is known for   small $\ti\al$  and for $\vP\neq 0$ in some ball $\mcS$ around zero \cite{HH08, CFP09}. On the other hand,
for any infrared cut-off $\si>0$ in the interaction, the resulting fiber Hamiltonians $H_{\vP,\si}$ do have (normalised)  ground states $\Psi_{\vP,\si}$
in the same region of parameters $\alt, \vP$.  Although these vectors tend weakly to zero as $\sigma\to 0$ \cite{CFP09}, they define states on a certain $C^*$-algebra $\mfa\subset B(\mcF)$:
\beqa
\om_{\pmb{P}}(A)=\lim_{\sigma\to 0}\lan \Psi_{\pmb{P},\sigma},A  \Psi_{\pmb{P}, \sigma}\ran, \quad A\in\mfa 
\label{original-Psi-0}.
\eeqa
These states can be interpreted as plane-wave configurations of the electron moving with momentum $\vP$.
It is well known that in (\ref{intro-Hamiltonian}), and in similar models of non-relativistic QED,  the GNS representations $\pi_{\pmb{P}}$
of the states $\om_{\pmb{P}}$  are disjoint for different values of $\pmb{P}\in \mcS$  \cite{Fr73, CF07, CFP09, KM14, CD18}. To our knowledge, 
this  mathematical formulation of velocity superselection was first introduced by Fr\"ohlich in \cite{Fr73}. In our recent work \cite{CD18}  
we showed that all the states $\{ \om_{\pmb{P}}\}_{\pmb{P}\in \mcS}$ belong to a suitably defined equivalence class, similar in intention to the charge classes from \cite{BR14}.
We also could resolve the velocity superselection by inserting certain \emph{infravacuum automorphisms} \cite{KPR77} between the `bare electron' and `soft-photon dressing'
constituting the states $\om_{\pmb{P}}$. In the present paper we cure velocity superselection in a more geometric manner, which we now briefly explain.   

It follows from the proof of Proposition~\ref{superselection-proposition} below that the  choice of the algebra $\mfa$ in (\ref{original-Psi-0}) is 
largely arbitrary, as long as it acts irreducibly on $\mcF$ and the states (\ref{original-Psi-0})  are well-defined.  In our paper we choose as $\mfa$
the algebra of observables of the free electromagnetic field. As this theory is local and relativistic, we have a subalgebra $\mfa(V_{+})\subset \mfa$
of the fields localised in the future lightcone. While $\pi_{\vP}$, $\pi_{\vP'}$ are disjoint as representations on the full algebra $\mfa$, we show that they are  unitarily equivalent after restriction to $\mfa(V_+)$. {\bc Actually we even show that $\pi_{\vP}$
are \emph{lightcone normal}, i.e., unitarily equivalent to the vacuum representation $\pi_{\mrm{vac}}$ after restriction to the lightcone. 
}

{\bc
Let us  explain the idea behind the proof  in heuristic terms:  Consider the formal expression
\beqa
W_{\vP}:=\exp\bigg( i\fr{\ti{\al}^{1/2}}{(2\pi)^{3/2}}\int_0^{\infty} dt\, \nabla E_{\vP}\cdot  \mathbf{A}(g)(-t-u, -\nabla E_{\vP}t )\bigg),
\label{first-time-integral}
\eeqa
where $E_{\vP}$ denotes the infimum of the spectrum of $H_{\vP}$, $g\in C_0^{\infty}(\real^3)$ is a smearing function, and $u>0$ is sufficiently large.
 The problem of convergence of the $t$-integral above will be left aside in this introductory discussion.  
 Up to a phase and the shift $u$, the expression $W_{\vP}$ is the incoming Dyson wave operator for the electromagnetic field interacting  with an external current. As expected, this current corresponds to an electron moving with velocity $\nabla E_{\vP}$,
 whose charge distribution is described by $g$, (cf. \cite[Section 6]{DH19}).  Using standard results from \cite{CF07, CFP09} on the states (\ref{original-Psi-0}), it is easily seen that their GNS representations  $\pi_{\vP}$ are unitarily equivalent to 
 $\pi_{\mrm{vac}}(W_{\vP}\,\cdot\, W_{\vP}^*)$.
  We  show that all $\pi_{\vP}$, $\vP\in \mcS$, are unitarily equivalent to 
 $\pi_{\mrm{vac}}$  by verifying that $W_{\vP}$ commute with $\mfa(V_+)$. 
If $\mathbf{A}$ was a local field, the  expression $W_{\vP}$ would
clearly be localized in the backward lightcone. Since this is not the case, we need one more step: using $\mathbf{E}=-\pa_t \mathbf{A}$
we express $\mathbf{A}$ as an integral of the free electric field $\mathbf{E}$, i.e.,
\beqa
W_{\vP} =\exp\bigg(-i\fr{\ti{\al}^{1/2}}{(2\pi)^{3/2}}\int_0^{\infty} dt\,\int_t^{\infty}d\tau\,  \nabla E_{\vP}\cdot  \mathbf{E}(g)(-\tau-u, -\nabla E_{\vP}t )\bigg), 
\label{double-time-integral}
\eeqa
which is manifestly localised in the backward lightcone $V_-$. Now by the Huyghens principle  $\mfa(V_-)\subset \mfa(V_+)'$
so we obtain lightcone normality of $\pi_{\vP}$.

The above intuitive arguments can be made rigorous by a careful control of the $t$-integrals in (\ref{double-time-integral}).
As this discussion is quite technical, we postpone it to Appendix~\ref{appendix-A}. In Section~\ref{sec:norma} we give
a less technical but also less insightful complex function argument, relying heavily on the fact that $\pi_{\mrm{vac}}(W_{\vP}\,\cdot\, W_{\vP}^*)$  is a coherent state. Although coherent states are well studied \cite{Ro70}, their behaviour
under restrictions to lightcones seems to be a `forgotten chapter', perhaps familiar to some experts but never  published.
We believe that there is a good reason to revisit this subject. Namely,  the general discussion of Buchholz and Roberts \cite{BR14} may suggest that all the coherent states from \cite{Ro70} with the usual infrared singularity are lightcone normal. 
We show in  Theorem~\ref{general-main-theorem} below that this is not the case. Specifically, consider functions of the form
\beqa
\vv(\vk) :=\fr{i}{|\vk|^{3/2}}  F(|\vk|)\vh(\hvk),
\eeqa
where $F: \real_+\to \complex$ is H\"older continuous at zero, $|\vk|^{-1/2}F$ is square-integrable outside  zero  and 
$\vh\in L^2(S^2;\real^3)$  is non-zero and transverse. Then coherent state representations given on Weyl operators
 by $\pi_{\vv}(W(\vf)):=e^{-2i\mrm{Im}\lan \vv, \vf\ran } \pi_{\mrm{vac}}(W(\vf))$ are lightcone normal
\emph{if and only if}
\beqa
\mrm{Im} F(0)=0. \label{reality-condition}
\eeqa
As coherent states $\pi_{\mrm{vac}}(W_{\vP}\,\cdot\, W_{\vP}^*)$ discussed above satisfy this condition,
this general result gives lightcone normality of the states $\pi_{\vP}$ and the absence of velocity superselection
on the lightcone algebra. In the same time,  Theorem~\ref{general-main-theorem} raises the question if 
coherent states with $\mrm{Im} F(0)\neq 0$ are relevant for infrared problems.  We  remark here that coherent states 
considered in \cite{DW19}, related to certain
gauge transformations in external current QED,  violate condition~(\ref{reality-condition}). However, as $\vh$ of (\ref{reality-condition}) is a distribution in this case,  Theorem~\ref{general-main-theorem} does not apply directly.
A further analysis of this issue, which is left for future research, may help to understand if different gauges can be
distinguished inside the future lightcone. This in turn may shed  light on the limitations of the Buchholz-Roberts approach.


}

  

\vspace{0.2cm}

\nin\textbf{Acknowledgements:} We would like to thank Detlev Buchholz and Henning Bostelmann for helpful discussions. We also thank Pawe\l{}  Duch for a useful hint in the second part of the proof of Lemma~\ref{lemma-symplectic-equivalence}.
{\bc Furthermore, we thank an anonymous referee for pointing out to us the computation~(\ref{long-computation})}.
This work was  supported by the DFG within the Emmy Noether grants DY107/2-1, DY107/2-2 and CA 1850/1-1.

\section{Free electromagnetic field}

We set $L^2(\real^3;\complex^3):=L^2(\real^3)\otimes \complex^3$   and denote the scalar product by $\lan \,\cdot\,, \, \cdot \, \ran$. The single-photon Hilbert space $\mfh$ is spanned by the transverse functions
\beqa
L^2_{\tr}(\real^3;\complex^3):=\{ \vf\in L^2(\real^3;\complex^3)\,|\,  \pmb{k} \cdot \pmb{f}(\pmb{k}) =0 \;\; \text{a.e.}\}
\eeqa 
and we denote by $P_{\tr}: L^2(\real^3;\complex^3)\to L^2(\real^3;\complex^3)$  the orthogonal projection on $L^2_{\tr}(\real^3;\complex^3)$.
We set $\hat{\pmb{k}}:=\pmb{k}/|\pmb{k}|$,  write $S^2$  for the unit sphere in $\real^3$
and introduce the polarisation vectors $S^2 \ni \hat{\vk}\mapsto  \pmb{\epsilon}_{\pm}(\hat{\vk})\in S^2$, given by, e.g., \cite{LL04}
\begin{eqnarray}
\pmb{\epsilon}_+ ( \hat{\pmb{k}})= \frac{(\hat{k}_2, -\hat{k}_1,0)}{\sqrt{\hat{k}_1^2 + \hat{k}_2^2}}, \quad
\pmb{\epsilon}_-( \hat{\pmb{k}})=   \hat{\pmb{k}}   \times \pmb{\epsilon}_+(  \hat{\pmb{k}}),
\end{eqnarray}
 which satisfy $\pmb{k} \cdot \pmb{\epsilon}_\pm(\hat{\vk}) =0$ and ${\pmb{\epsilon}}_+(\hat{\vk}) \cdot \pmb{\epsilon}_-(\hat{\vk}) =0$ for $\hat{\vk}=(\hat{k}_1,\hat{k}_2,\hat{k}_3)\in S^2$. With the help of these vectors we can write
\beqa
(P_{\tr}\vf)(\vk)=\sum_{\la=\pm }  \big(\vf(\vk)\cdot  \pmb{\epsilon}_\la (\hat{\pmb{k}} ) \big) \pmb{\epsilon}_\la ( \hat{\pmb{k}}) \label{transverse-projection}
\eeqa
and note that the right hand side  of the latter equality is actually meaningful for any function $\vf: \real^3\to \complex^3$. For a given choice of the polarisation
vectors we can identify $L^2_{\tr}(\real^3; \complex^3)$ with $L^2(\real^3;\complex^2)$ via 
\beqa
L^2_{\tr}(\real^3; \complex^3)\ni \fb \mapsto (f_{+}, f_-)\in L^2(\real^3;\complex^2), \quad  f_{\pm}:=\pmb{\epsilon}_{\pm}   \cdot \fb.
\eeqa
Next, we denote
by $\mcF$ the symmetric Fock space over $ \mfh:=L^2_{\tr}(\real^3; \complex^3)\simeq L^2(\real^3; \complex^2)$:
\begin{equation}\label{Fock-space}
{\mathcal{F}} := \oplus_{n=0}^\infty {\mathcal{F}}^{(n)}, \quad {\mathcal{F}}^{(n)} :=  \operatorname{Sym}_n (  \mfh  \phantom{}^{\otimes n}),\quad 
{\mathcal{F}}^{(0)} = \mathbb{C}\Omega.
\end{equation}
The dense domain of finite particle vectors will be denoted by $\mcF_0$ and $\mathcal{D}_{S}\subset \mcF_0$ will denote the subspace of
finite particle vectors with Schwartz-class wave functions.

Let $a^{(*)}(\,\cdot\,)$ be the creation and annihilation operators on this Fock space and $a^{(*)}_{\la}(\vk)$ the improper creation and annihilation
operators on $\mcF$ such that $[a_{\la}(\vk), a^*_{\la'}(\vk')] =\de_{\la\la'}\de(\vk-\vk')$  and all other commutators vanish. These operators are related by
$a^*(\vf)=\sum_{\la=\pm} \int d^3\pmb{k}\, a_{\la}^*(\vk)\, (\pol(\hat\vk)\cdot \vf(\vk))$,  for $\vf\in \mfh$.  

Now we define the electromagnetic potential in the Coulomb gauge as an operator valued distribution on Fock space\footnote{We skip the usual normalisation constant 
$\fr{1}{(2\pi)^{3/2}} \fr{1}{\sqrt{2}} $ for consistency with \cite{CFP09}.}
\beqa
\vA(t,\vx):=\sum_{\la=\pm }\int \fr{d^3\pmb{k}}{\sqrt{|\vk|} }\, \pmb{\epsilon}_{\la}(\hat\vk)\big( e^{i|\vk|t-i\vk\cdot \vx}a^*_{\la}(\vk)+   e^{-i|\vk|t+i\vk\cdot \vx}a_{\la}(\vk)  \big). \label{em-potential}
\eeqa
More precisely, for any $\fb\in D(\real^4;\real^3)$,  (the space of smooth, compactly supported functions from $\real^4$ to $\real^3$),  the expression
\beqa
\vA(\fb):=\int dtd^3\vx\, \vA(t,\vx)\!\cdot\! \fb(t,\vx)
\eeqa
defines  an essentially self-adjoint operator on $\mcF_0$, whose self-adjoint extension will be denoted by the same symbol (cf.~\cite[Section X.7]{RS2}).  The same applies to
the electromagnetic fields, which are defined as distributions by
\beqa
\mathbf{E}(t,\vx)=-\pa_{t} \vA(t,\vx), \quad \mathbf{B}(t, \vx)=\mrm{rot}\, \vA(t,\vx).
\eeqa 

In contrast to the electromagnetic potential above, the electromagnetic fields are Wightman fields. They give rise to a Haag-Kastler net of local $C^*$-algebras
which is constructed in a standard manner: For any double cone\footnote{A double cone is a spacetime translate of a set 
$\mco_r:=\{\,(t,\pmb{x})\in \real^4\,|\, |t|+|\pmb{x}|<r\}$, $r>0$. We also say that $O_r:=\{\, \pmb{x}\in \real^3 \,|\,|\pmb{x}|<r\}$ is the base of $\mco_r$.} $\mco\subset \real^4$ we define 
the local algebra $\mfa(\mco)$ as the $C^*$-algebra generated by exponentials of the smeared fields:
\beqa
\mfa(\mco):=C^*\{  e^{i (\mathbf{E}(\vfe)+\mathbf{B}(\vfb)) } \,|\,  \supp\, \vfe, \supp\,\vfb\subset \mco\,\}. \label{local-algebra}
\eeqa
The algebras associated with any (possibly unbounded) open regions $\mathcal{U}$ are obtained by the $C^*$-inductive limit, i.e.,
\beqa
\mfa(\mathcal{U}):=\ov{\bigcup_{\mco\subset \mathcal{U}}  \mfa(\mco)}^{\|\,\cdot\, \|}.
\eeqa
This gives, in particular, the quasi-local algebra $\mfa:=\mfa(\real^4)$ and the algebras $\mfa(V_{\pm})$ of the
future (+) and backward (-) open lightcone with a tip at zero.

The net of algebras $\mco\mapsto \mfa(\mco)$ is local, i.e., $\mfa(\mco_1)\subset \mfa(\mco_2)'$, where   $\mco_1$ and $\mco_2$
are spacelike-separated and the prime denotes the commutant in $B(\mcF)$. Even more importantly, the
Huyghens principle holds, that is,
\beqa
\mfa(V_-)\subset \mfa(V_+)'. \label{Huyghens}
\eeqa
We will also use the translation covariance property, which gives
\beqa
e^{iH_{\mrm{ph}}t-i\vP_{\mrm{ph}}\cdot \vx    } \mfa(\mco)e^{-iH_{\mrm{ph}}t+i\vP_{\mrm{ph}}\cdot \vx }=\mfa(\mco+(t,\vx)),
\eeqa
where the energy-momentum operators
\beqa
H_{\mrm{ph}}:=\sum_{ \la=\pm} \int d^3\pmb{k}\, |\vk|\, a^*_{\la}(\vk) a_{\la}(\vk), \quad \vP_{\mrm{ph}}:=\sum_{ \la=\pm} \int d^3\pmb{k}\, \vk\, a^*_{\la}(\vk) a_{\la}(\vk)
\eeqa
are essentially self-adjoint on  $\mathcal{D}_{S}$ and their self-adjoint extensions are denoted by the same symbol.

It will be convenient to express the algebras above as CCR algebras in the Fock representation. For this purpose, for any $\vfe, \vfb\in D(\real^4;\real^3)$, we write
\beqa
\vf(\vk):= -i(2\pi)^2  \bigg( |\vk|^{1/2} P_{\mrm{tr}} \ti\vfe(|\vk|,\vk) + |\vk|^{-1/2}(\vk\times \ti\vfb(|\vk|,\vk) )\bigg), \label{fourier-representation}
\eeqa
where tilde denotes the Fourier transform\footnote{We use the conventions for the Fourier transform from \cite{RS2}, i.e., $\tilde{f}(k^0, \pmb{k})
=\fr{1}{(2\pi)^2} \int dt d^3 \pmb{x}\, e^{ik^0t-i\pmb{k}\cdot \pmb{x} } f(t,\pmb{x}).$ }. We define the real-linear vector spaces
\beqa
\mcL(\mco):=\{ \vf \,|\, \supp\, \vfe, \supp\, \vfb\subset \mco\,\}, \quad \mcL(\mathcal{U}):=\bigcup_{\mco\subset \mathcal{U}} \mcL(\mco) \label{symplectic-spaces}
\eeqa
and equip them with the symplectic form $\mathbf{\si}(\vf_1, \vf_2)=\mrm{Im}\lan \vf_1, \vf_2\ran$. Then $W(\vf):=e^{i(a^*(\vf  )+a(\vf))}$
satisfy the Weyl relations
\beqa
W(\vf_1) W(\vf_2)=e^{-i \mathbf{\si}(\vf_1, \vf_2) } W(\vf_1+\vf_2), \quad W(\vf)^*=W(-\vf).
\eeqa
 We note  that 
$\mfa(\mco)=\mrm{CCR}(\mcL(\mco))$, $\mfa=\mrm{CCR}(\mcL)$ and $\mfa(V_{\pm})= \mrm{CCR}(\mcL(V_{\pm}))$, where
 $\mrm{CCR(\ti\mcL)}$ denotes the $C^*$-algebra generated by $W(\vf)$, $\vf\in \ti\mcL$.
Since $\mcL:=\mcL(\real^4)$ is dense in $L^2_{\mrm{tr}}(\real^3;\complex^3)$,  the quasi-local algebra $\mfa$ acts irreducibly on $\mcF$.  {\bc The defining  representation of $\mfa$ will be denoted $\pi_{\vac}$. It is the GNS representation of
the vacuum state $\om_{\mrm{vac}}(\, \cdot\,):=\lan \Om, \, \cdot \,\, \Om\ran$. We say that a given representation $\pi$ of
$\mfa$ is \emph{lightcone normal} if
\beqa
\pi  \restriction \mfa(V_+)\simeq  \pi_\vac \restriction \mfa(V_+),
\eeqa
where $\simeq$ denotes the unitary equivalence. Lightcone normality of states is defined w.r.t. their GNS representations.}

\section{Lightcone normality of coherent states}\label{sec:norma}

In this section we investigate {\bc lightcone normality} of coherent states on the algebra of the free electromagnetic field. 
{\bc Although coherent states have been well studied, e.g. by Roepstorff \cite{Ro70}, we are not aware of any treatment of this particular aspect in the literature. Theorem~\ref{general-main-theorem} below gives an exact characterization of 
lightcone normality for coherent states with the usual infrared singularity. Our analysis reveals that the lightcone normality for
such states  is not automatic but requires an additional `reality assumption' on the defining functional.  }


Consider functions of the form
\beqa\label{v}
\vv(\vk) :=\fr{i}{|\vk|^{3/2}}  F(|\vk|)\vh(\hvk). \label{general-coherent-functions}
\eeqa
Here $\vh\in L^2_{\mrm{tr}}(S^2, \real^3)$, {\bc $\vh \neq 0$},  and $F: \real_+\to \complex$ is a measurable function which satisfies
\begin{enumerate}
\item[(a)] H\"older continuity at zero, i.e., $|F(0)-F(|\vk|)| \leq c |\vk|^{\eps}$ for some $\eps>0$ and all $|\vk|\leq 1$,

\item[(b)] $\int_{\si}^{\infty} d|\vk| |\vk|^2 \big|\fr{1}{|\vk|^{3/2}}  F(|\vk|)\big|^2 <\infty$ {\bc \textrm{ for any } $\si>0$}.

\end{enumerate}
As $\vv$ are not square-integrable, it is convenient to introduce an approximating sequence of 
$L^2$-functions:
\beqa
\vv_{\si}(\vk)=\fr{i}{|\vk|^{3/2}}\chi_{[\si, \infty)}(|\vk|) F(|\vk|)\vh(\hvk),
\eeqa  
where $\chi_{\De}$ is the characteristic function of a set $\De$. Now we consider coherent automorphisms of $\mfa$ defined by
\beqa\label{exp}
\alphan_{\vv}(W(\fb)) := \lim_{\si \to 0}W(\vv_{\si})W(\fb) W(\vv_{\si})^*=  e^{-2i \im  \langle \vv, \fb \rangle} W(\fb).
\eeqa
The main result of this section is  now a characterization of lightcone normality of the coherent states above under mild regularity conditions.
\bet\label{general-main-theorem} For $\vv$ as in (\ref{general-coherent-functions}), satisfying properties (a), (b), 
we have
\beqa
 \pi_\vac \circ \alphan_{\vv}  \restriction \mfa(V_+)\simeq  \pi_\vac \restriction \mfa(V_+),
\eeqa
if and only if $\mrm{Im} F(0)= 0$. 
\eet
\begin{proof} {\bc First  suppose  that $\mrm{Im} F(0)=0$}.  
{\bc Then we can find a real-valued function  $G\in C_0^{\infty}(\real)$, supported in the interior of the negative real axis, s.t.
$\ti{G}(0)=F(0)$
and define an auxiliary  function}
\beqa
\hat\vv(\vk)=\fr{i}{|\vk|^{3/2}}  \ti{G}(|\vk|)\boldsymbol{h}(\hvk). \label{hat-v}
\eeqa
It follows from properties (a), (b) and $F(0)=\ti G(0)$ that $\vv- \hat\vv\in L^2_{\mrm{tr}}(\real^3; \complex^3)$.
Thus, by standard arguments (e.g. Lemma 1 of \cite{Ro70}), {\bc $\alphan_{\vv} =\mrm{Ad}U\circ  \alphan_{\hat{\vv}}$}
for some unitary $U$. Therefore, to conclude the proof of {\bc the if-part of} Theorem~\ref{general-main-theorem}, it suffices to show that 
$\alphan_{\hat{\vv}}\restriction \mfa(V_+)=\mrm{id} \restriction \mfa(V_+)$, {\bc where $\mrm{id}$ is the identity mapping}. 
This is a consequence {\bc of (\ref{exp}) and  Lemma~\ref{Strocchi-Lemma} below.}

{\bc Now suppose that $\mrm{Im} F(0)\neq 0$}. {\bc Then $F=\mrm{Re}F+i \mrm{Im}F$ gives the corresponding 
decomposition $\vv=\vv_0-\check\vv$, where $\vv_0$ is as in the first part of the proof and}
\beqa\label{check-v}
\check\vv(\vk):=\fr{1}{|\vk|^{3/2}} ({\bc\mrm{Im}\,F})(|\vk|)\boldsymbol{h}(\hvk).
\eeqa
{\bc Let us show  that $\alphan_{\check\vv}$ is not lightcone normal by the method of central sequences}. Suppose, {\bc by contradiction}, that $\pi_{\vac}\circ \alphan_{\check\vv} \restriction \mfa(V_+)\simeq \pi_{\vac} \restriction \mfa(V_+)$ for some unitary $U$ on  $\mcF$. This implies
\begin{equation}\label{inner}
\omega_\vac\circ \alphan_{\check\vv} \restriction \mfa(V_+) = \omega_\vac \circ \mrm{Ad}U \restriction \mfa(V_+).
\end{equation}
Since the representation $\pi_{\vac}$ of $ \mfa$ is irreducible, then by \cite[Theorem~10.2.1]{KR} there exists $\tilde U \in \mfa$ such that $\tilde U^\ast \Omega = U^\ast \Omega$. Thus we can write, for all $A\in  \mfa(V_+) $,
\begin{equation}\label{comm}
\omega_\vac( \alphan_{\check\vv} (  A)) = \omega_{\vac}( \tilde U  A \tilde U^\ast) = \omega_\vac( \tilde U [ A, \tilde U^\ast]) +  \omega_\vac(  A).
\end{equation}
Now let $\vf\in \mcL(V_+)$. We introduce the sequence $\vf_\lam (\vk):=\lam^{-3/2}\vf(\vk/\lam) $, $\lam>0$, 
and set $A=W(\vf_\lam)$ in \eqref{comm}.  By Lemma~\ref{first-lemma}, we have
\begin{equation}
\lim_{\lam \to 0}\big( e^{-2i \mrm{Im} \lan \check{\vv} , \vf_\lam   \ran} - 1 \big) \omega_\vac(  W(\vf_\lam)) =0.
\end{equation}
Since $\omega_\vac( W(\vf_\lam))$ is constant in $\lam$ {\bc and non-zero}, and since we could always change $\vf_\lam$ by a multiplicative constant, we arrive at a contradiction, considering Lemma~\ref{second-lemma} below. 

{\bc Finally, recalling  that $\alphan_{\check\vv} =\alphan_{\vv_0} \circ \alphan_{-\vv}$ and that $\alphan_{\vv_0}$ is
lightcone normal by the first part of the proof, we conclude that $\alphan_{\vv}$ cannot be lightcone normal if  $\mrm{Im} F(0)\neq 0$.}
\end{proof}
\bel\label{Strocchi-Lemma}  For $\hat\vv$ as in (\ref{hat-v}),  {\bc with $G\in C_0^{\infty}(\real)$  real-valued and supported in the interior of the negative real axis}, and any $\vf\in \mcL(V_+)$
\beqa\label{scalarprod}
 \mrm{Im} \lan \hat{\vv} , \vf   \ran=0.  
\eeqa 
\eel
\proof 
We  write
$\vf(\vk):= -i(2\pi)^2  \bigg( |\vk|^{1/2} P_{\mrm{tr}} \ti\vfe(|\vk|,\vk) + |\vk|^{-1/2}(\vk\times \ti\vfb(|\vk|,\vk) )\bigg)=:\vf^{\mrm{e}}+\vf^{\mrm{b}}$
and consider the resulting two contributions:
\beqa
& &\mrm{Im} \lan \hat\vv, \vf^{\mrm{e}}\ran
= -\fr{(2\pi)^2}{2i}  \int d\Om(\hvk) \vh(\hvk)\big(\int_0^{\infty} d\rho \, \rho  \ti{G}(-\rho)   \ti{\vfe}(\rho,\hvk \rho)\non\\
& & \ph{44444444444444444444444444444444}+ \int_{-\infty}^{0} d\rho \,   \rho\ti{G}(-\rho)   \ti{\vfe}(\rho, \hvk \rho)\big)\non\\
& &\ph{444444444}=-\fr{(2\pi)^2}{2i}  \int d\Om(\hvk) \vh(\hvk)  \int_{\real} d\rho \, \rho  \ti{G}(-\rho)   \ti{\vfe}(\rho,\hvk \rho).
\label{long-computation}
\eeqa
Now we want to close the contour in {\bc the} upper complex half-plane. We write $z=\rho+i\eta$ and note 
\beqa
\ti{\vfe}(z,\hvk z) 
\!\!\!&=&\!\!\!  \fr{1}{(2\pi)^2} \int e^{ (i\rho-\eta) (t- \hvk\cdot \vx)   } \vfe(t,\vx) dt d^3\vx.
\eeqa
Since $\vf$ is supported in the future lightcone, we have $t- \hvk\cdot \vx\geq \de>0$, uniformly in $\hvk$. Hence
\beqa
|\ti{\vfe}(z,\hvk z)|\leq C e^{-\eta \de}.
\eeqa 
Moreover, considering that $G$ is supported in $(-\infty, -R]$, for some $R>0$, we have
\beqa
|\ti{G}(-z ) |\leq Ce^{-R\eta}.
\eeqa
Thus we can close the contour and conclude that $\mrm{Im} \lan \hat{\vv}, \vf^{\mrm{e}}\ran=0$. To show that 
$\mrm{Im} \lan \hat{\vv}, \vf^{\mrm{b}}\ran=0$ we proceed analogously. \qed

\begin{lemma} \label{first-lemma}
For $C  \in \mfa$ {\bc and $\vf_\lam$ defined below (\ref{comm})} we have $\lim_{\lam \to 0}\| [C, W(\vf_\lam)] \| = 0$.
\end{lemma}
\begin{proof}
For any $\eps$ we can find an observable $C_{\eps}$ localized in some double cone $\mco^{(\eps)}$ {\bc such that
$\|C-C_{\eps}\|\leq \eps$.} Thus we can write
\begin{align}
 |[C,    W(\pmb{f}_{\lam}) ]\| \leq &2\|C-C_{\eps}\|\|W(\pmb{f}_{\lam})\|+   \|[C_{\eps},  W(\pmb{f}_{\lam}) ]\|\non\\
 \leq & 2\eps + \|[C_{\eps},  W(\pmb{f}_{\lam}) ]\|.
  \end{align}
Now it suffices to show that for any fixed $\eps$ we have $\lim_{\lam\to 0} \|[C_{\eps},  W(\pmb{f}_{\lam}) ]\|=0$.
For this purpose, we note that $W(\pmb{f}_{\lam})\in \mfa(V_+ + \lam e_+)$, with some $e_+ \in V_+$.  
For sufficiently large $\lambda$, the vector $\lambda e_+$ is in the future of $\mco^{(\eps)}$ and thus $ \|[C_{\eps},  W(\pmb{f}_{\lam}) ]\|=0$ by the Huyghens principle \bc{(\ref{Huyghens}}). This concludes the proof.
\end{proof}

\begin{lemma} \label{second-lemma}
{\bc Let $\check\vv(\vk):=\fr{1}{|\vk|^{3/2}} F(|\vk|)\boldsymbol{h}(\hvk)$, where $F, \vh$ are as in  (\ref{general-coherent-functions})
and in addition $F(0)$ is non-zero and real.
Then,  for any $\vh\neq 0$ there exists $\vf\in \mcL(V_+)$ such that}
\beqa
\lim_{\lam \to 0} \mrm{Im} \lan \check{\vv} , \vf_\lam   \ran \neq 0.  
\eeqa 
\end{lemma}
\begin{proof}
Proceeding as in computation (\ref{long-computation}) we get
\beqa
& &\mrm{Im} \lan \check\vv, \vf^{\mrm{e}}\ran=- \fr{(2\pi)^2}{2i}  \int d\Om(\hvk) \vh(\hvk)\big(\int_0^{\infty} d\rho \, i\rho  \ov{F}(\rho)   \ti{\vfe}(\rho,\hvk \rho)\non\\
& & \ph{44444444444444444444444444444444}- \int_{-\infty}^{0} d\rho \,  i \rho F(-\rho)   \ti{\vfe}(\rho, \hvk \rho)\big).
\eeqa
Setting $\vf_{\mrm{e},\lam}(t, \vx) := \lam^{2}\vf_{\mrm{e}}(\lam t, \lam \vx)$ {\bc we obtain} $\ti{\vf}_{\mrm{e},\lam}(|\vk|, \vk)=\lam^{-2}\ti\vf_{\mrm{e}}(|\vk|/\lam, \vk/\lam)$
and 
\beqa\label{rhoplus}
& &\mrm{Im} \lan \check\vv, \vf^{\mrm{e}}_{\lam}\ran=- \fr{(2\pi)^2}{2i}  \int d\Om(\hvk) \vh(\hvk)\big(\int_0^{\infty} d\rho \, i\rho  \ov{F}(\lam \rho )   \ti{\vfe}(\rho,\hvk \rho )\non\\
& & \ph{44444444444444444444444444444444}- \int_{-\infty}^{0} d\rho \,  i \rho F(-\lam \rho)   \ti{\vfe}(\rho , \hvk \rho)\big).
\eeqa
{\bc Now we use that {\bc $F(0)$ is real} to get in the limit $\lam \to 0$}
\beqa\label{sequence}
& &\lim_{\lam\to 0}\mrm{Im} \lan \check\vv, \vf^{\mrm{e}}_{\lam}\ran=- (2\pi)^2  F(0)\int d\Om(\hvk) \vh(\hvk)\big(\int_{0}^{\infty} d\rho \, \rho\, \mrm{Re}\big(  \ti{\vfe}(\rho,\hvk \rho ) \big) \big). 
\eeqa
{\bc In the remaining part of the proof} we will exhibit $\vfe \in \mcL(V_+)$ such that  (\ref{sequence}) is different from zero
thus proving our claim with $\ti\vfb =0$.

The integration in $\rho$ {\bc in (\ref{sequence})} can be computed by means of Fourier transforms of distributions \cite[Sec.~3]{GS}, and its real part yields 
\begin{equation}
 \mrm{Re}\big(\int_{0}^{\infty} d\rho \, \rho\,\ti{\vfe}(\rho,\hvk \rho ) \big) =-\fr{1}{(2\pi)^2} \int_{V_+} dt d^3\vx\, \vfe(t,\vx) (t- \hvk\cdot \vx +i0 )^{-2}.
\end{equation}
Note that since $\vfe$ is supported inside the {\bc future} light cone, there is no singularity at $t- \hvk\cdot \vx=0$ in the above expression, and the regularization there can be dropped. Inserting into \eqref{sequence}, we have
\begin{equation}
\lim_{\lam\to 0}\mrm{Im} \lan \check\vv, \vf^{\mrm{e}}_{\lam}\ran= F(0)\int d\Om(\hvk) \vh(\hvk) \int_{V_+} dt d^3\vx\, \vfe(t,\vx) (t- \hvk\cdot \vx )^{-2},
\end{equation}
Now, since the function 
\begin{equation}\label{fct}
{\bc V_+ \ni (t, \vx) }  \mapsto \int d\Om(\hvk) \vh(\hvk)  (t- \hvk\cdot \vx )^{-2}
\end{equation}
is analytic in $(t, \vx) \in V_+$ as the integration region is compact, then it either does not vanish except on a null set, or it vanishes identically. In the first  case, one can find $\vfe$ such that \eqref{sequence} is non-zero, finishing the proof. In the second case, we will construct a contradiction. We expand the function $\lambda \mapsto  (t- \lambda \hvk\cdot \vx )^{-2}$ around $\lambda =0$, and obtain
\begin{equation}
 \int d\Om(\hvk) \vh(\hvk)  (t- \lambda \hvk\cdot \vx )^{-2} = \sum_{\ell =0}^\infty (\ell+1)t^{-\ell-{\bc 2}} \lambda^\ell \int d\Om(\hvk) \vh(\hvk)  (\hvk\cdot \vx )^\ell.
\end{equation}
If the r.h.s. vanishes identically for all $(t, \vx) \in V_+$, it follows that
\begin{equation}
\forall \ell,\vx  \; :\;  \int d\Om(\hvk) \vh(\hvk) (\hvk\cdot \vx )^\ell =0.
\end{equation}
Now for $\vx = \pmb{e}_3$, the unit vector in the direction of the $z$-axis, we have in usual spherical coordinates,
\begin{equation}
\forall \ell,j  \; :\;  \int d\Om(\hvk) h_j (\hvk) \cos^\ell \theta =0,
\end{equation}
thus $h_j$ is orthogonal to all $Y_{\ell 0}$. 

Since the representation of the rotation group is irreducible at every fixed angular momentum $\ell$, we can use our choice of $\vx $ to show orthogonality to all rotated $Y_{\ell 0}$, and therefore to all $Y_{\ell m}$. Thus $ \vh \equiv 0$ {\bc which is a contradiction}.
\end{proof}

\section{Pauli-Fierz model of non-relativistic QED}\label{sec:pauli}
Our aim is to apply {\bc the} results from the previous section in the Pauli-Fierz model of non-relativistic QED. We now summarize some known facts about this model, as used in \cite[Subsection~4.1]{CD18}. 
By analogy with (\ref{em-potential}), 
we define  the quantized electromagnetic vector potential
with infrared and ultraviolet cut-offs $0\leq \si\leq \ka$  
as the following  operator on  $\mcF_0$
\beqa
\vA_{[\si,\ka]}(\vx):=\sum_{\la=\pm} \int \fr{d^3\pmb{k}}{\sqrt{|\vk|}} \chi_{[\si, \ka]}(|\vk|) \pol(\hat\vk)\big( e^{-i\vk\cdot \vx} a^*_{\la}(\vk) +   e^{i\vk\cdot \vx} a_{\la}(\vk)  \big), \label{fiber-Hamiltonians}
\eeqa
where $\chi_{\De}$ denotes the characteristic function of a set $\De$.  The fiber Hamiltonians from the decomposition (\ref{fiber-hamiltonians}) are given by  
\beqa
H_{\vP,\si}=\h(\vP-\vP_{\pho}+ \alt^{1/2}\vA_{[\si,\ka]}(0))^2+H_{\pho},\quad H_{\vP}:=H_{\vP,\si=0}. \label{Hamiltonian-section1}
\eeqa
They are self-adjoint, positive operators on a domain in $\mcF$, which is independent of $\vP$ (see, e.g., \cite{Sp, Hi00, KM14}). 
The infima of the spectra of $H_{\vP,\si}$, $H_{\vP}$, denoted by $E_{\vP,\si}:=\mrm{inf}\,\mrm{ Spec}(H_{\vP,\si})$, $E_{\vP}:=\mrm{inf}\,\mrm{ Spec}(H_{\vP})$ 
 are rotation invariant functions of $\vP$.

Now we recall some spectral results, mostly from \cite{CFP09, FP10},   which will be used in the next section. From now on we discuss
the regime of low coupling $\ti\al >0$ and 
momenta $\vP$ restricted to the ball
\beqa
\mcS=\Big\{ \vP\in \real^3\,|\, |\vP|< \fr{1}{3} \Big\}.
\eeqa 
It is well  known that for any $\si>0$ the operators $H_{\vP,\si}$ have ground-states $\Psi_{\vP,\si}\in \mcF$, $\|\Psi_{\vP,\si}\|=1$, so that
  $E_{\vP,\si}$ are eigenvalues. The dependence $\vP\mapsto E_{\vP,\si}$ is analytic for any fixed $\si>0$
by the Kato perturbation theory.   In the limit $\si\to 0$ the vectors $\Psi_{\vP,\si}$  tend weakly to zero 
\cite{CFP09,  Fr73, Fr74.1, Ch00} and the Hamiltonians $H_{\vP}$ do not have ground-states for $\vP\neq 0$ \cite{HH08}.  
To analyze  this phenomenon, one introduces the auxiliary vectors 
\beqa
\Phi_{\vP,\si}:=W(-i\vv_{\vP, \si})\Psi_{\vP,\si}, \quad W(-i\vv_{\vP,\si})=e^{a^*( \vv_{\vP,\si})-a(\vv_{\vP,\si})   }, \label{modified-states}
\eeqa
where $\vv_{\vP,\si}$ has the form
\beqa
\pmb{v}_{\pmb{P},\sigma}(\pmb{k}) =   \alt^{1/2}  P_{\tr} \frac{\chi_{[\sigma, \kappa]} (|\pmb{k}|) }{|\pmb{k}|^{3/2}}
\frac{\nabla E_{\pmb{P},\sigma}}{1  - \hat{\pmb{k}}\cdot \nabla E_{\pmb{P},\sigma}}, \label{v-P-sigma}
\eeqa
and we set $\hat{\vk}:=\vk/|\vk|$ and $\nabla E_{\pmb{P},\sigma}:=\nabla_{\vP}E_{\pmb{P},\sigma}$. (By a slight abuse of notation, we use in (\ref{modified-states}) the notation $W(\vf)$ also for $\vf$ which are not in the spaces (\ref{symplectic-spaces})). The following lemma collects some facts from \cite{CFP09, FP10}.\footnote{Precisely, for (a) and (b) see \cite[Theorem III.3 and Corollary III.4]{FP10}, for (c) see \cite[Eq.~(III.2) and  formula (V.6)]{CFP09} and for (d) \cite[Theorem III.1]{CFP09}.}
\bel\label{spectral} Let $\ti\al>0$ be sufficiently small and $\vP\in \mcS$. Then
\begin{enumerate}
\item[(a)] The function $\vP \mapsto E_{\vP}$ is rotation invariant, twice differentiable and has a strictly positive second derivative with respect to $|\vP|$. 
\item[(b)] $\lim_{\si\to 0} \pa_{\vP}^{\be} E_{\vP,\si}$ exists and equals $\pa_{\vP}^{\be} E_{\vP}$ for $|\be|\leq 2$.
\item[(c)] $|\nabla E_{\vP,\si}|\leq v_{\mathrm{max}}<1$ and $|\nabla E_{\vP}|\leq v_{\mathrm{max}}<1$ for some constant $v_{\mathrm{max}}$, uniformly in $\sigma$ and in $\vP \in \mcS$.
\item[(d)] $\Phi_{\vP}:=\lim_{\si\to 0}\Phi_{\vP,\si}$ exists in norm for a suitable choice of the phases of $\Psi_{\vP,\si}$. 
\end{enumerate}
\eel
\nin In the following we assume that the phases of $\Psi_{\vP,\si}$ are fixed as in Lemma~\ref{spectral}~(d). Using Lemma~\ref{spectral}~(b) 
we can define the pointwise limit
\beqa
\vv_{\vP}(\vk):=\lim_{\si\to 0}\vv_{\vP,\si}(\vk)= {\alt^{1/2}  P_{\tr}  \frac{\chi_{[0, \kappa]} (|\pmb{k}|) }{|\pmb{k}|^{3/2}}
\frac{ \nabla E_{\pmb{P}} }{1  - \hat{\pmb{k}}\cdot \nabla E_{\pmb{P}}}.} \label{v-P}
\eeqa
We note that the expressions $1-\hat{\pmb{k}}\cdot \nabla E_{\pmb{P}, \si}$ and  $ 1-\hat{\pmb{k}}\cdot \nabla E_{\pmb{P}}$ 
in the denominators of (\ref{v-P-sigma}) and (\ref{v-P}) are different from zero by Lemma~\ref{spectral}~(c).  Furthermore,
 $P_{\mrm{tr}}$  acting in (\ref{v-P}) on a function which is not in $L^2(\real^3;\complex^3)$ is defined by the right hand side of (\ref{transverse-projection}). 
The fact that   $\vv_{\vP}$ is not in $L^2_{\tr}(\real^3; \complex^3)$ for $0\neq \vP\in \mcS$ will be important below.  

\section{Curing velocity superselection}\label{sec:curing}

Now let us consider a special example of the state \eqref{v} which is relevant in the Pauli-Fierz model and related to the problem of velocity superselection. On the CCR algebra $\mfa$ over the symplectic space 
$\mcL$ as introduced above, we define
\beqa
\om_{\pmb{P}}(A):=\lim_{\sigma\to 0}\lan \Psi_{\pmb{P},\sigma}, A  \Psi_{\pmb{P}, \sigma}\ran=\lan \Phi_{\pmb{P}},
{\bc \al_{-i\vv_{\pmb{P} }  } }(A)  \Phi_{\pmb{P}}\ran,
\quad A\in\mfa, \label{state-formula}
\eeqa
{\bc where the automorphism $\al_{-i\vp}$ is defined as in (\ref{exp})}
{\bc and $\vps$ is given by (\ref{v-P-sigma})}.
These states describe plane-wave configurations of the electron with velocity~$\nabla E_{\vP}$. 
Now let $\pi_{\vP}$ be the GNS representation of $\om_{\vP}$. By formula~(\ref{state-formula}) and standard arguments (see, e.g., \cite[Lemma A.1]{CD18}), we have
\beqa
\pi_{  \pmb{P} }\simeq \pi_{\vac}\circ \al_{  {\bc -i\vv_{\pmb{P}} }  }, \label{representation-equation}
\eeqa
where $\pi_{\vac}$ is the defining Fock vacuum representation and $\simeq$ denotes unitary equivalence. Thus, in particular, $\pi_{  \pmb{P} }$ are irreducible representations.  

The mathematical formulation of velocity superselection, consisting in the disjointness of $\pi_{  \pmb{P} }$ for distinct $\pmb{P}$,
was introduced by Fr\"ohlich in \cite{Fr73}   and established later by various authors in different models and for
varying choices of the algebra $\mfa$  \cite{CFP09, CF07, Fr73,KM14, CD18}. 
{\bc From the argument below it is clear that the details of the construction of $\mfa$ are largely arbitrary.}
\bep\label{superselection-proposition} Let  $\vP, \vP'\in \mcS$,  $\vP\neq \vP'$.  Then $\pi_{\vP}$ and $\pi_{\vP'}$ are disjoint.
\eep
\begin{proof}  
{\bc We adapt Lemma 1 of \cite{Ro70}}. Suppose by contradiction that there is a unitary $U$ such that
\beqa
\pi_{\vac} \circ \al_{ {\bc -i\vv_{\pmb{P}} } }=\mrm{Ad}U\, \circ \pi_{\vac} \circ \al_{ {\bc -i \vv_{\pmb{P}' } }} \Rightarrow 
\pi_{\vac} \circ \al_{ {\bc-i(\vv_{\pmb{P}} - \vv_{\pmb{P}' }) }  }=\mrm{Ad}U\, \circ \pi_{\vac}.  \label{implication}
\eeqa
{\bc Since $\mcL$ is dense and $\vv_{\pmb{P}} - \vv_{\pmb{P}' }$ is not square-integrable, we can find
a sequence $\mcL\ni \vf_n\to 0$ in $L^2$ s.t.
\beqa
\lim_{n\to \infty}\mrm{Im}\lan i(\vv_{\pmb{P}} - \vv_{\pmb{P}' }), \vf_n\ran \neq 0.
\eeqa
By evaluating both sides of the second relation in (\ref{implication}) on $W(\vf_n)$ and 
using that $W(\vf_n)\to I$ in the strong operator topology, we conclude the proof. }
\end{proof}
{\bc The main result of this section} is the following theorem,  which says that velocity superselection can be resolved by restriction to the future lightcone.
\bet\label{prop:vp} For any $\vP, \vP'\in \mcS$ we have $\pi_{\vP}\restriction \mfa(V_+)\simeq \pi_{\vP'} \restriction \mfa(V_+)$.
\eet
\nin {\bc In view of (\ref{representation-equation})}, this is a consequence of Theorem~\ref{general-main-theorem} above.
A different argument, which is applicable only to {\bc representations $\pi_{\boldsymbol{P}}$}, is given in Appendix~\ref{appendix-A}.

\section{Conclusions}

In this paper we showed that the problem of velocity superselection of the electron can be
resolved by restriction to the algebra of the future lightcone  $V_+$.   
We considered only the lightcone with a tip  at zero, but a generalisation to shifted  lightcones is
straightforward. As expected from the time-reversal symmetry of QED, restriction to a  backward lightcone $V_-+\mathbf{a}$, $\mathbf{a}\in \real^4$, has the same effect. 
We showed that the GNS representations of $\mfa(V_+)$ are unitarily equivalent for a large class of coherent states, of which those in the Pauli-Fierz model are an example.
We are confident that
analogous results hold in other models of non-relativistic QED by suitably adapted arguments. For example in the Nelson model,
which describes the electron interacting with the massless scalar field, already a counterpart of  (\ref{first-time-integral}) would give 
an approximating sequence localised in the backward lightcone, and a double-integral formula (\ref{double-time-integral}) would not be needed.

Proceeding towards future research directions, we recall that there is a more satisfactory concept of velocity superselection
in non-relativistic QED, which uses the algebra  generated by  \cite{CFP07}
\beqa
W_{\mrm{out}}(\pmb{h}):=e^{i(a^*_{\mrm{out}}(\pmb{h})+a_{\mrm{out}}(\pmb{h})) },  \textrm{ where } a_{\mrm{out}}^*(\pmb{h})=\lim_{t\to\infty} e^{itH}a^*(e^{-it|\pmb{k}|} \pmb{h})  e^{-itH}
\eeqa
and  $h$ are suitable functions.
 The representations induced by the infraparticle scattering states on this
algebra have a direct integral decomposition into disjoint representations labelled by the electron's asymptotic velocity. We conjecture
that also in this context the algebra of the future lightcone can be found, on which these representations are unitarily equivalent. 
Such analysis may pave the way to suitably dressed Hamiltonians of non-relativistic QED, for which the infraparticle problem disappears.
We hope to come back to this problem in a future investigation.

\appendix


\section{{\bc An alternative proof of Theorem~\ref{prop:vp}  } } \label{appendix-A}

Theorem~\ref{prop:vp} follows immediately from Lemmas~\ref{lemma1} and \ref{lemma2} below.

%
{\bc First, we  introduce an auxiliary function $\vpt$  given by}
\begin{equation}\label{vpt-equation}
\vpt :=\alt^{1/2} P_{\text{tr}}\frac{ \tilde g(\kb) e^{-iu|\kb|}\nabla E_{\Pb}}{\km^{3/2} (1 - \nabla E_{\Pb} \cdot \hat{\kb})},
\end{equation}
{\bc where} $g:\mathbb{R}^3 \to \mathbb{R}$ is a smooth function with compact support, with $\tilde g(\pmb{0})= 1$, and \textcolor{black}{$u >1$} is so large that $(-u,\mathbf{x})\in V_-$ for all $\mathbf{x}\in\operatorname{supp} g$. 
Since $\vp-\vpt\in L^2_{\mrm{tr}}(\real^3;\complex^3)$, we have by standard arguments (e.g. Lemma 1 of \cite{Ro70}): 
\begin{lemma}\label{lemma1}
$\alpha_{{\bc -i\vp}} \circ \alpha_{ {\bc-i \vpt}}^{-1}$ acts by the adjoint action of a unitary  on the $C^\ast$-algebra  ${\bc \mfa}$. 
\end{lemma} 
%
%
{\bc \nin Now it suffices to prove the following:}
\begin{lemma}\label{lemma2}
The automorphism $\alpha_{{\bc -i\vpt}}$ acts like the identity on $\mathfrak{A}(V_+)$.
\end{lemma}
\begin{proof}
We only need to show that $\alpha_{{\bc -i\vpt}}(W(\fb))=W(\fb)$ for all $\fb \in \mathcal{L}(\mathcal{O})$ and $\mco\subset V_+$. As remarked {\bc in the Introduction}, this is achieved by approximating $\vpt$ with functions localized in the (standard) backward light cone, and using timelike commutativity of the free electromagnetic field. Hence we define the approximant, $T>0$,
\begin{equation}\label{cutoff}
(\vptT{})_{\lambda}(\kb) :=  -\alt^{1/2} \int_0^T dt  \int_t^{T} d\tau\; \sqrt{\km} \nabla E_{\Pb} \cdot   \pmb{\epsilon}_\lambda(\hat\kb) \tilde g(\kb)  e^{-i \km u}  e^{-i(\km \tau - \nabla E_{\Pb} \cdot \kb t)}.
\end{equation}
This suggests an approximating sequence for $W(-i\hat{\pmb{v}}_{\pmb{P}  })$,
\begin{align}
W(-i\hat{\pmb{v}}_{\pmb{P},T  }) &= 
 \exp\Big( - \alt^{1/2} \sum_{ \lambda=\pm }\int_0^T dt\;  \int_t^{T} d\tau \, \int d^3\kb \non\\  & \, (\sqrt{\km} \nabla E_{\Pb} \cdot   \pmb{\epsilon}_\lambda(\hat\kb) \tilde g(\kb)  e^{-i \km u}  e^{-i(\km \tau - \nabla E_{\Pb} \cdot \kb t)} a_\lambda^{ *} (\kb) )  -\text{h.c.}\Big) \non \\
 &= \exp\Big( -i \alt^{1/2} \int_0^T dt\;  \int_t^{T} d\tau \,  \fr{1}{(2\pi)^{3/2}} \nabla E_{\Pb} \cdot \pmb{E}(g) \big( -u-\tau,   -\nabla E_{\Pb}t  \big) \Big), 
 \label{lightcone-line}
\end{align}
considering that with our conventions,
\begin{align}
\pmb{E}(t,\pmb{x})=-\sum_{  \lambda=\pm} \int d^3\pmb{k}\, \sqrt{|\pmb{k}| } \pmb{\epsilon}_{\lambda}(\hat{\pmb{k}}) i 
\big(  e^{i|\pmb{k}|t-i\pmb{k}\cdot \pmb{x} } 
a_{\lambda}^{*}(\pmb{k})-   e^{-i |\pmb{k}| t +i\pmb{k}\cdot \pmb{x} }   a_{\lambda}(\pmb{k})\big).
\end{align}

The region of integration in (\ref{lightcone-line}) is depicted in Figure~\ref{fig:region}. 
As remarked above, $u$ is chosen so large that $-i\hat{\pmb{v}}_{\pmb{P},T  }$ is contained in $\mathcal{L}(V_-)$ and thus  $W(-i\hat{\pmb{v}}_{\pmb{P},T  })\in \mfa(V_-)$ (see Lemma~\ref{Lemma-in-appendix-x}). Therefore $\im \langle -i\vptT, \fb \rangle =0$ if  $\fb \in \mathcal{L}(\mathcal{O}) \subset \mathcal{L}(V_+)$, see \eqref{Huyghens}.
\begin{figure}
\begin{center}
 \includegraphics[width=0.95\textwidth]{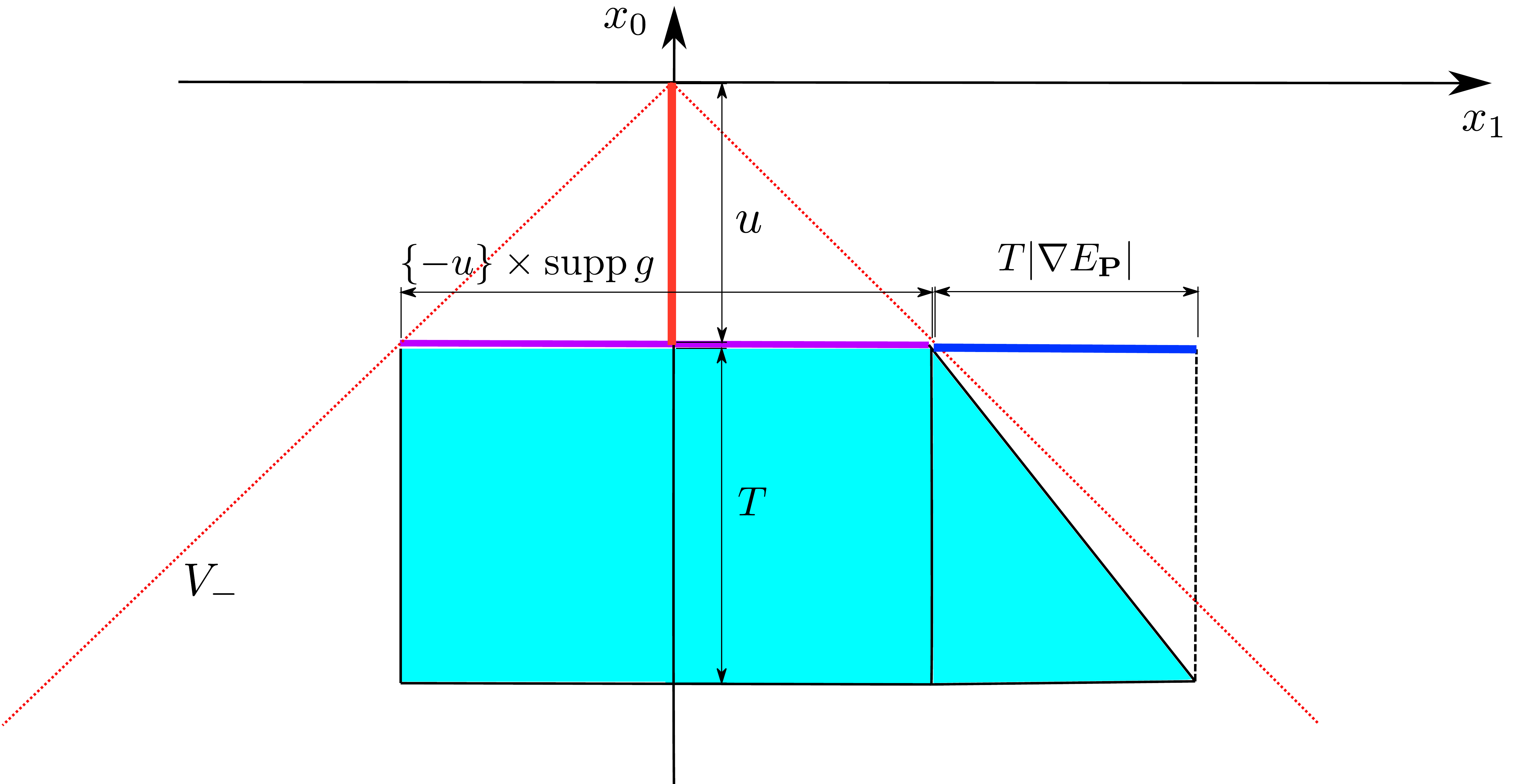}
 \end{center}
 \caption{Integration region  in Eq.~(\ref{lightcone-line}).} \label{fig:region}
\end{figure}
It now suffices to check that $\lim_{T\to\infty}\langle \vptT, \fb \rangle = \langle \vpt, \fb \rangle$ for all $\fb\in \mathcal{L}(\mathcal{O})$.  Then
 $\im\langle  -i\vpt, \fb \rangle=0$ and hence $\alpha_{{\bc -i\vpt}}(W(\fb))=W(\fb)$ by its definition \eqref{exp}.

To that end, we first perform the $\tau$- and $t$-integrations in \eqref{cutoff}, which give
\begin{equation}\label{3terms}
\begin{aligned}
(\vptT{})_\lambda(\kb) &=  (\vpt)_{\la}(\kb)\\ 
&- \alt^{1/2} \nabla E_{\Pb} \cdot   \pmb{\epsilon}_\lambda(\hat{\kb}) \tilde g(\kb)  e^{-i \km u}  e^{-i\km T}\frac{1}{\km^{3/2}}\frac{1}{ \nabla E_{\Pb} \cdot \hat{\kb} }\Big\lbrack   e^{i \nabla E_{\Pb} \cdot \kb T} -1 \Big\rbrack \\
&-\alt^{1/2} \nabla E_{\Pb} \cdot   \pmb{\epsilon}_\lambda(\hat{\kb}) \tilde g(\kb)  \frac{1}{\km^{3/2}(1-\nabla E_{\Pb} \cdot \hat{\kb} )}   e^{-i \km u}e^{-i( \km - \nabla E_{\Pb} \cdot \kb) T} .
\end{aligned}
\end{equation}
We need to show that the last two  terms in \eqref{3terms} vanish weakly in the limit $T \to \infty$. The last of these terms gives a contribution to 
$\ov{ \langle \vptT, \fb \rangle}$ of
\begin{equation}\label{term2}
-\sum_{ \lambda=\pm }\int d^3\kb\, \alt^{1/2}  \nabla E_{\Pb} \cdot   \pmb{\epsilon}_\lambda(\hat{\kb}) \tilde g(\kb) \frac{ 1}{\km^{3/2}(1-\nabla E_{\Pb} \cdot \hat{\kb})} e^{-i \km u} e^{-i( \km - \nabla E_{\Pb} \cdot \kb) T} \overline{f_\lambda(\kb)}.
\end{equation}
This vanishes in the limit $T\to \infty$ due to the dominated convergence for the angular integration in $d\Omega(\hat{\kb})$, and by applying the Riemann-Lebesgue lemma to the one-dimensional integration in $d\km$ with oscillating factor $e^{-i\km(1 - \nabla E_{\Pb}\cdot \hat{\kb})T}$. For the relevant majorants, note that $\tilde g$ is Schwartz and that the integrand behaves like $\km^{-3/2}$ at small $\kb$, which remains integrable with respect to $d^3\kb$.

The second term of  \eqref{3terms}  gives the contribution 
\begin{multline}
 -\sum_{ \lambda=\pm}  \int d^3\kb\, \alt^{1/2}  \nabla E_{\Pb} \cdot   \pmb{\epsilon}_\lambda(\hat{\kb}) \tilde g(\kb)  e^{-i \km u}  e^{-i\km T}\frac{1}{\km^{3/2}}\frac{1}{\nabla E_{\Pb} \cdot \hat{\kb} }\Big\lbrack   e^{i \nabla E_{\Pb} \cdot \kb T} -1 \Big\rbrack  \ov{f_\lambda(\kb) } \\
=  -i\alt^{1/2}T\sum_{ \lambda=\pm } \int_0^1 d\beta\, \int d\Omega(\hat{\kb}) \int_0^\infty d\km\,     \nabla E_{\Pb} \cdot   \pmb{\epsilon}_\lambda(\hat{\kb}) \tilde g(\kb)  e^{-i \km u} \km^{3/2}  \ov{f_\lambda(\kb)} e^{-i\km T(1- \beta \hat{\kb}\cdot \nabla E_{\Pb})}.
\end{multline}
Integrating by parts twice in $\km$ we obtain: 
\begin{equation}\label{termi3}
\frac{i \textcolor{black}{\alt^{1/2}}}{T} \sum_{ \lambda=\pm }\int_0^1 d\beta\, \int d\Omega(\hat{\kb}) \int_0^\infty d\km\, \frac{ 
\frac{\partial^2}{\partial \km^2}\Big\lbrack \nabla E_{\Pb} \cdot   \pmb{\epsilon}_\lambda(\hat{\kb}) \tilde g(\kb)  e^{-i \km u} \km^{3/2} \ov{f_\lambda(\kb)}  \Big\rbrack 
}{
(1- \beta \hat{\kb}\cdot \nabla E_{\Pb})^2
}
e^{-i\km T(1- \beta \hat{\kb}\cdot \nabla E_{\Pb})}\end{equation}
up to boundary terms which vanish for any fixed $\hat{\kb}$ since $\km^{3/2}f_\lambda (\kb)$ vanishes as $\km \to 0$ together with its derivative with respect to $\km$
 (cf.~Eq.~\eqref{fourier-representation}),  and  since $\pmb{\epsilon}_\lambda(\hat{\kb})$ are chosen independent of $\km$.
We estimate the above integral as follows:
\begin{equation}\label{fin}
\lvert \eqref{termi3} \rvert \leq \frac{ \textcolor{black}{\textcolor{black}{2}\alt^{1/2} u^2}\lvert \nabla E_{\Pb} \rvert}{Tc^2} \sum_{ \lambda=\pm } \int d\Omega(\hat{\kb})  \int_0^\infty d\km\,  \sum_{\ell=0,1,2} \Big\lvert \frac{\partial^{\ell}}{\partial \km^{\ell}}\Big\lbrack  \km^{3/2}  \tilde g(\kb)   \ov{f_\lambda(\kb)}  \Big\rbrack   \Big\rvert
\end{equation}
using that $1- \beta \hat{\kb}\cdot \nabla E_{\Pb} \geq 1- \lvert \beta  \rvert\lvert \hat{\kb} \rvert \lvert \nabla E_{\Pb} \rvert =: c$. 
Taking into account that $\fb\in \mathcal{L}(\mathcal{O})$ and that $\tilde g$ is Schwartz, one finds that the second derivative is integrable in $\km$ with a bound for the integral uniform in $\hat\kb$.
Hence the  integrals are all finite, and \eqref{fin} vanishes in the limit $T \to \infty$. 
\end{proof}


\section{Equivalence of two definitions of the symplectic space}

Lemma~\ref{Lemma-in-appendix-x} from this appendix is used in the proof of Lemma~\ref{lemma2} above.

Let $O_r\subset \real^3$ be an open ball of radius $r$ centered at zero and let
$J$ be the complex conjugation in configuration space. Following \cite{BuJa} we define the
symplectic space
\beqa
& &\mcL_{\mrm{BJ}}:=\bigcup_{r>0}  \mcL_{\mrm{BJ}}(O_r), \quad \textrm{where}\\
& &\mcL_{\mrm{BJ}}(O_r):= (1+J)\ov{  |\vk|^{-1/2} (i \pmb{k} \times  \wt{D}(O_r; \real^3))  }+(1-J)\ov{ |\vk|^{1/2} P_{\mrm{tr}}  \wt{D}(O_r; \real^3)  }.
\eeqa
We recall that the spaces $\mcL(\mco)$ and the symplectic space $\mcL$ were defined in  (\ref{symplectic-spaces}) and note the following
lemma. (A similar discussion of the scalar field can be found in \cite[Section 7.4.1]{Bo00}).
\bel\label{lemma-symplectic-equivalence} For any $r>0$ we have $\mcL_{\mrm{BJ}}(O_r)=\mcL(\mco_r)$
where $\mco_r$ is the double cone centered at zero whose base is $O_r$. Hence, $\mcL_{\mrm{BJ}}=\mcL$.
\eel
\proof In order to show $\mcL(\mco_r)   \subset  \mcL_{\mrm{BJ}}(O_r)$, we decompose  
$\vf$ given in (\ref{fourier-representation}) into its real and imaginary part in configuration space
\beqa
\vf=\fr{(1+J)}{2} \vf +\fr{(1-J)}{2} \vf. \label{J-decomposition}
\eeqa
Next, exploiting that $P_{\text{tr}} \vf=  -\frac{1}{\km^2} \kb \times \lbrack \kb \times  \vf \rbrack$, we obtain
\begin{align}
\fr{(1+J)}{2}\vf(\pmb{k})&=(-i) (2\pi)^2 \km^{-1/2} \kb \times  \Big\lbrack \kb \times \Big( - \frac{\ti\vfe(|\vk|,\vk) - \overline{\ti\vfe(|\vk|,-\vk)}}{2 \km}  \Big) \nonumber \\ 
&\quad\quad\quad\quad\quad\quad\ph{444444444444444444444444}  +\frac{\ti\vfb(|\vk|,\vk) + \overline{\ti\vfb(|\vk|,-\vk)}}{2} \Big\rbrack.
\end{align}
It is easy to see that
\begin{align}
 \frac{\ti\vfe(|\vk|,\vk) - \overline{\ti\vfe(|\vk|,-\vk)}}{2 \km}  &= \frac{i}{(2\pi)^2} \int dt d^3\vx\, \vfe(t,\vx) e^{-i \kb \cdot \vx} \frac{\sin (\km t)}{\km}, \label{sine-one}\\
\frac{\ti\vfb(|\vk|,\vk) + \overline{\ti\vfb(|\vk|,-\vk)}}{2}  &= \frac{1}{(2\pi)^2} \int dt d^3\vx\, \vfb(t,\vx) e^{-i \kb \cdot \vx} \cos (\km t). \label{cosine-one}
\end{align}
The rapid decay of (\ref{sine-one}) and (\ref{cosine-one}) as $|\pmb{k}|\to \infty$  implies smoothness of their inverse Fourier transforms. 
By choosing the polar coordinates, we compute  the inverse Fourier transform of \eqref{cosine-one}:
\begin{align}
&\frac{1}{(2\pi)^2} \int dt d^3\vx\, \vf_{\mrm{b}}(t,\vx) \int d^3 \kb\, e^{-i \kb \cdot (\vx - \pmb{y})} \cos (\km t) \label{inverse-fourier-transform}\\
&\ph{4444444444444444444444}= 4\pi \frac{1}{(2\pi)^2} \int dt d^3\vx\, \frac{\vf_{  \mrm{b}}  (t,\vx)}{\lvert \vx - \pmb{y} \rvert} \int_{ 0}^{\infty} d\km\, \km  \cos (\km t) \sin (\km \lvert \vx - \pmb{y}\rvert) \nonumber\\
&\ph{4444444444444444444444} =  \frac{1}{\pi} \int dt d^3\vx\, \frac{\vf_{ \mrm{b}}(t,\vx)}{\lvert \vx - \pmb{y} \rvert} \delta'(t - \lvert  \vx - \pmb{y}\rvert).\nonumber
\end{align}
By this formula, $\supp \vf_{\mrm{b} } \subset \mco_r$ implies that the expression in (\ref{inverse-fourier-transform}) is supported in $O_r$ in the $\pmb{y}$
variable. An analogous argument applies to (\ref{sine-one}). Also, 
the analysis of the second term on the right hand side of (\ref{J-decomposition}) follows the same steps.

To justify $\mcL(\mco_r)\supset  \mcL_{\mrm{BJ}}(O_r)$, we choose an arbitrary $f_{i}\in D(O_r;\real)$ and consider a smooth solution of the wave
equation of the form
\beqa
g_i(t,\vx)=\fr{1}{(2\pi)^{3/2}}\int d^3\vk\, e^{i\vk\cdot \vx} \fr{\sin(|\vk|t)  }{|\vk|}   \ti f_i(\vk),  
\eeqa
which is compactly supported in space for any fixed $t$ and satisfies $g_i(0,\vx)=0$, $(\pa_t g_i)(0,\vx)=f_i(\vx)$. Thus we can 
write
\begin{align} \label{smearing-symplectic-form}
\int d^3\vx\, E_i(0, \vx)f_i(\vx)&=\int d^3\vx\, \big( E_i(0, \vx) (\pa_t g_i)(0,\vx)-(\pa_t  E_i)(0, \vx) g_i(0,\vx) \big)\non\\
&= \int d^3\vx\,  E_i(t, \vx) \overset{\leftrightarrow}{\pa}_t g_i(t,\vx)=\int d\tau \,\al(\tau)\int d^3\vx\,  E_i(\tau, \vx) \overset{\leftrightarrow}{\pa}_{\tau} g_i(\tau,\vx),
\end{align}
where in the last step we made use of the time-invariance of the symplectic form on the space of solutions of the wave equation to integrate
with $\al \in D(\real;\real)$ such that $\int d\tau \, \al(\tau)=1$, whose support is chosen in a sufficiently small neighbourhood of zero. Considering 
that an analogous equality holds for the components of the magnetic
field and acting with  both sides of (\ref{smearing-symplectic-form}) on the vacuum, we conclude from the finite propagation speed of $g_i$ that   $\mcL(\mco_r) \supset  \mcL_{\mrm{BJ}}(O_r)$. \qed\\
As an application of Lemma~\ref{lemma-symplectic-equivalence}, we show that the expression in (\ref{lightcone-line}) is an element 
of the $C^*$-algebra $\mfa(V_-)$ (and not only of its weak closure).
\bel\label{Lemma-in-appendix-x} In the notation from the proof of Lemma~\ref{lemma2}, we have 
\beqa
W(-i\vptT)\!=\!\exp\Big(\fr{ \!-\! i \alt^{1/2}}{(2\pi)^{3/2}} \int_0^T dt  \int_t^{T} d\tau\; \nabla E_{\Pb} \cdot \pmb{E}(g) \big( -u-\tau,  { -}\nabla E_{\Pb}t  \big) \Big)
\! \in\! \mfa(V_-).\,\,\,
\eeqa
\eel
\proof We note the equality
\beqa
& &e^{i|k|(u+T)}(-i\vptT{})= \int_0^T dt  \int_t^{T} d\tau\;  \pmb{v}_\mathrm{int}(\tau,t),  \label{shifted-integral}\\
& &   \pmb{v}_\mathrm{int}(\tau,t):=  i \alt^{1/2} \sqrt{\km} P_{\text{tr}}\nabla E_{\Pb} \tilde g(\kb)e^{i \km T}  e^{-i(\km \tau - \nabla E_{\Pb} \cdot \kb t)}.
\eeqa
We recall that $u> 1$ is chosen so large that $\supp\, g\subset O_{u}$. 
Following the steps from the proof of Lemma~\ref{lemma-symplectic-equivalence}, one can show that the integral on the right hand side of 
(\ref{shifted-integral}) belongs to $\mcL_{\mrm{BJ}}(O_{u+T})$. Considering this,
by Lemma~\ref{lemma-symplectic-equivalence} it belongs to $\mcL(\mco_{u+T})$, where $\mco_{u+T}$ is the double
cone whose base is $O_{u+T}$.  Then, by equality (\ref{shifted-integral}), $-i\vptT{}\in \mcL(V_-)$. \qed


\end{document}